\documentclass[aps,twocolumn,pra,superscriptaddress,letterpaper]{revtex4-1}
\usepackage{amsmath,amsfonts,amssymb,amsthm}
\usepackage{mathtools}
\usepackage[pdftex]{graphicx}
\usepackage[usenames, dvipsnames, svgnames, table]{xcolor}
\usepackage[colorlinks=true, citecolor=Blue, linkcolor=BrickRed, urlcolor=Brown]{hyperref}
\usepackage{txfonts}
\usepackage{mathrsfs}
\usepackage{bm}
\usepackage{multirow}
\usepackage{braket}
\usepackage{import}
\usepackage{qcircuit}
\usepackage[ruled, noline, noend, algo2e]{algorithm2e}
\usepackage{array}
\usepackage{comment}
\usepackage{tikz}
\usepackage{algorithm}

\newtheorem{theorem}{Theorem}
\newtheorem{lemma}{Lemma}

\newtheorem{corollary}{Corollary}
\theoremstyle{definition}
\newtheorem{definition}{Definition}
\newcommand{\be}{\begin{equation}}
\newcommand{\ee}{\end{equation}}
\newcommand{\ben}{\begin{eqnarray}}
\newcommand{\een}{\end{eqnarray}}
\newcommand{\bes}{\begin{subequations}}
\newcommand{\ees}{\end{subequations}}
\newcommand{\bF}{\begin{figure}}
\newcommand{\eF}{\end{figure}}

\DeclareMathAlphabet{\pazocal}{OMS}{zplm}{m}{n}

\newcommand{\orcid}[1]{\href{https://orcid.org/#1}{\includegraphics[height = 2ex]{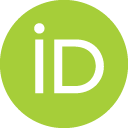}}}

\begin{document}

\title{Extensively Not P-Bi-Immune promiseBQP-Complete Languages}

\author{Andrew Jackson \orcid{0000-0002-598one-one604}}
\affiliation{Department of Physics, University of Warwick, Coventry CV4 7AL, United Kingdom}
\date{\today}


\begin{abstract}
    In this paper, I first establish - via methods other than the Gottesman-Knill theorem - the existence of an infinite set of instances of simulating a quantum circuit to decide a decision problem that can be simulated classically. I then examine under what restrictions on quantum circuits the existence of infinitely many classically simulable instances persists.
    There turns out to be a vast number of such restrictions, and any combination of those found can be applied at the same time without eliminating the infinite set of classically simulable instances.
    Further analysis of the tools used in this then shows there exists a language that every (promise) BQP language is one-one reducible to. This language is also not P-bi-immune under very many promises.
\end{abstract}
\maketitle
\section{Introduction}
\subsection{Background}
The relationship between BQP~\cite{doi:10.1137/S0097539796300921} and P is perhaps the central question in quantum computer science. With billions being invested in quantum computing~\cite{Bogobowicz_Gao_Masiowski_Mohr_Soller_Zemmel_Zesko_2023} around the world~\cite[Fig. 2]{whitePaper}, in the expectation of a large payoff~\cite{Bova2021}, ``are we wasting our time?" becomes a vital question. Because if $BQP = P$, then there is no advantage to developing quantum computers. However, deciding if this is the case is, unfortunately, an immensely difficult problem. Great efforts have been directed towards a solution to this problem~\cite{doi:10.1137/S0097539796300921, doi:10.1137/S0097539795293639, 10.1145/1806689.1806711, classSimIsPolyCollapse, 10.1145/1993636.1993682}, and while some progress has been made (most notably the existence of an oracle~\footnote{Although oracle results can be misleading~\cite{10.1145/146585.146609, CHANG199424}} relative to which BQP $\not \subseteq$ PH~\cite{10.1145/3313276.3316315} and similarly there are oracles relative to which BQP $\not =$ BPP~\cite{doi:10.1137/S0097539796300921, doi:10.1137/S0097539796298637, doi:10.1137/S0036144598347011}), a definitive answer to the question seems some way off\footnote{Especially as classical equivalents of quantum algorithms are occasionally found\cite{10.1145/3313276.3316310, Tang2022, PhysRevA.107.012421, gall2023robust}, showing that we do not perfectly understand the source of quantum advantage.}.

Given we cannot, in this paper or the near-term, determine the exact relationship between P and BQP, I instead target more easily determined properties of the relationship. Herein, I focus mostly on the question of P-bi-immunity (defined formally in Def.~\ref{def:definingBiImmunity}) in BQP-complete languages. I first establish the existence of an infinite set of instances of simulating a quantum circuit to decide a decision problem that can also be simulated classically.

This property of BQP is independent of the exact relation of BQP with P (i.e. resolving this question is not known to resolve if BQP = P) but is still a worthwhile result to establish.
Although not previously directly addressed, this question touches on the question of what conditions can we impose upon a circuit to make it classically simulable. The most celebrated instance is when we demand the circuit be Clifford\footnote{I.e. For every gate, $g$, in the circuit, for any string of Pauli gates, $P$, $gP g^{-1}$ is a string of Pauli gates.}, in which case the circuit is classically simulable~\cite{gottesman1998heisenberg, Aaronson_2004, Anders_2006}. This work is perhaps best seen as an expansion of this result, but in the abstract. then I show a vast and diverse array of restrictions we can place on the quantum circuits under which the existence of infinitely many classically simulable instances persists.
Notably, any combination of these restrictions can be applied simultaneously without eliminating the infinite set of classically simulable instances.
I additionally show that there is a one-one reduction to the problem of circuit simulation where the above results apply from every (promise) BQP language.


\subsection{Initial Definitions}
\subsubsection{Basic Objects and Operations}
Before diving into the main technical results of this paper, I first introduce the main objects and concepts studied herein. These objects form the very basics of computer science and are the most basic objects that theoretical computer science can be said to ``be about."
\begin{definition}
    An \underline{alphabet} is a finite set of distinct symbols
\end{definition}
Throughout this paper, alphabets will generally be denoted by $\Sigma$.

\begin{definition}
    Let $\Sigma$ be an alphabet, then $\Sigma^*$ is the set of all finite ordered sets, $x$, such that every element in $x$ is also in $\Sigma$.
    The elements in $x$ need not be distinct.\\
    $\Sigma^*$ is the result of the \underline{Kleene star operation} on $\Sigma$.
\end{definition}

\begin{definition}
    x is a \underline{word} over $\Sigma$ (an alphabet) $\iff$ $x \in \Sigma^*$
\end{definition}

\begin{definition}
    $\mathcal{L}$ is a \underline{language over $\Sigma$} $\iff$ $\mathcal{L} \in \mathcal{P}\big( \Sigma^* \big)$\\
    Where $\mathcal{P}\big( \Sigma^* \big)$ denotes the power set of $\Sigma^*$.
    I.e. a language is a set of words over an alphabet.
\end{definition}

\subsubsection{Decision Problems and Decision Quantum Circuit Simulation}
I now present further definitions, specifically about decision problems and a particular problem - decision quantum circuit simulation - that will be useful later.
\begin{definition}
     A \underline{decision problem} is a problem that can only have one of two answers (that can be interpreted as ``yes" and ``no" or ``accept" and ``reject.").
    For the purposes of this paper, a decision problem will exclusively be restricted to deciding if a given word is in a given language.
\end{definition}

\begin{definition}
    A \underline{promise} is a subset of $\Sigma^*$, for some alphabet $\Sigma$
\end{definition}

\begin{definition}
    A \underline{promise problem} is a decision problem but with the guarantee all inputs will in in a given promise.
\end{definition}

\begin{definition}
    A \underline{gateset of a quantum device} is the finite set of all the gates that can be applied on the quantum device. 
\end{definition}

\begin{definition}
An instance of \underline{Decision Quantum Circuit Simulation} (hereafter DCS) is defined by three objects:\\
$\bullet$ A binary input string, $s$\\
$\bullet$ An ordered set, $\mathbb{G}$, of unitary gates (from a fixed gateset) acting on $\textit{Poly} \big(\vert s \vert \big)$ qubits~\footnote{Where  $\textit{Poly} \big(\vert s \vert \big)$ qubits denotes that there exists some polynomial - a, potentially, different one for each instance of the problem bounding the number of qubits.}.\\
$\bullet$ An ordered set of single-qubit measurements, $\mathbb{M}$\\
These three objects define a quantum circuit: the input string defines a $Z$-basis product state prepared at the start of the circuit, the ordered set of gates in $\mathbb{G}$ are then applied to the prepared state before the measurements of $\mathbb{M}$ are applied.\\
The task is to return the most likely outcome of the last measurement in $\mathbb{M}$, subject to the promise the probability of each outcome is in $[0, \frac{1}{3}] \cup [\frac{2}{3}, 1]$.
\end{definition}
    Despite the usefulness of DCS, a specific input to it can be expressed more concisely by slightly changing the problem itself, to give Simplified Decision Quantum Circuit Simulation - as defined in Def.~\ref{def:sDCS}.
\begin{definition}
\label{def:sDCS}
An instance of \underline{Simplified Decision Quantum Circuit Simulation} (hereafter sDCS) is defined by a single object:\\
$\bullet$ An ordered set, $\mathbb{G}$, of gates (in a fixed gateset) acting on $\vert s \vert$ qubits.\\
The task is then to decide an instance of DCS with an input string $s$ being $0^{n}$ (where $n$ is the number of qubits defined by the gates in $\mathbb{G}$), the ordered set being $\mathbb{G}$, and $\mathbb{M}$ being fixed as single-qubit measurements in the computational basis on each qubit at the end of the circuit.
\end{definition}
There are in fact many versions of sDCS depending on the gateset used, but I generalised by not specifying the gateset and instead denoting it by $\mathbb{T}$.
There exists a one-one mapping from sDCS to $\Sigma^*$ - as shown in Algorithm~\ref{11MappingAlgorithm} - and so I can define a language $\mathcal{L}_{\textit{sDCS}} \subset \Sigma^*$ of instances where the measurement outcome with a probability of at least $2/3$ is $\vert 1 \rangle$.

\begin{algorithm}[H]
$\mathbf{Input:}$ \\
$\bullet$ A circuit, expressed in a specific gate set (of size, $\Xi \in \mathbb{N}$).
\begin{enumerate}
    \item  $\Sigma = \{z \text{ } \vert \text{ } z \in \mathbb{N} \text{ } s.t. \text{ } z \leq \vert \mathbb{T} \vert + 1 \}$
     \item Arbitrarily choose a pairing between elements of the gate set and the integers strictly greater than 1 (expressed in a base that is the same as the size of the gate set), call this $\bold{T}$.

     \item Arbitrarily choose a similar mapping between the each qubit and a subset of strings of a fixed length (the upper bound of the $\text{log}_{\Xi}$ of the number of qubits). Call this one-one mapping $\textit{Bin}$.   

\item stringToReturn = an empty string

\item For each gate $g$ in $\mathbb{G}$: \begin{enumerate}
    \item Let $g_{\textit{op}}$ be the actual operation applied e.g. a Pauli $X$ gate, and $g_{\textit{ind}}$ be the index of the qubit it is applied to.
    \item stringToReturn.append($\bold{T}(g_{\textit{op}})$)
    \item stringToReturn.append($\textit{Bin}(g_{\textit{ind}})$)
    \end{enumerate}
    \end{enumerate}
    $\mathbf{Return}:$ stringToReturn
    \caption{Algorithm for encoding a circuit in a string \label{11MappingAlgorithm}}
\end{algorithm} 
I also present an algorithm (Algorithm~\ref{circuitExtractingAlgorithm}), that extracts the encoded circuit from the string encoding it.
\begin{algorithm}[H]
$\mathbf{Input:}$ \\
$\bullet$ A string, $x$, encoding a circuit, on $N$ qubits, in a specified and known gate set.
\begin{enumerate}
    \item Split the string into substrings of $\lceil \text{log}_{\Xi} (N) \rceil + 1$ characters.
    \item circuitToReturn = empty circuit
    \item For every substring
    \begin{enumerate}
        \item Apply $\mathbb{T}^{-1}$ to its first character of the substring - label the result $g_{\textit{op}}$ -  and $\textit{Bin}^{-1}$ to the remaining characters - labelling the result $g_{\textit{ind}}$.
        \item Add $g_{\textit{op}}$ to the qubit $g_{\textit{ind}}$ in circuitToReturn.
    \end{enumerate}
\end{enumerate}
    $\mathbf{Return}:$ circuitToReturn
    \caption{Algorithm for extracting an encoded circuit in a string \label{circuitExtractingAlgorithm}}
\end{algorithm}

\subsubsection{Complexity Classes}
I now turn to complexity theory, to give a brief introduction to the important complexity classes that are relevant to the results in this paper.
\begin{definition}
    A problem is \underline{decidable} if it can be answered correctly by some Turing machine in finite time for any input.
\end{definition}

\begin{definition}
    \underline{P} is the set of all problems where the runtime to decide any input instance, on a deterministic Turing machine, is bounded by a polynomial of the size of the input.
\end{definition}

\begin{definition}
    \label{def:BQPviaQuantumTuring}
    \underline{BQP} is the set of all problems where the runtime to correctly decide any input instance with probability greater than $\frac{2}{3}$, on a quantum Turing machine, is bounded by a polynomial of the size of the input.
\end{definition}
Def.~\ref{def:BQPviaQuantumTuring} can be expressed differently - in Def.~\ref{def:BQPvia circuits} - using Lemma~\ref{lem:quantumTuringToCircuits}.
\begin{lemma}[Ref.~\cite{doi:10.1098/rspa.1985.0070}]
    \label{lem:quantumTuringToCircuits}
    Any decision problem decidable by a quantum Turing machine in polynomial (in the input size) time - with at least a $2/3$ probability of success - is also decidable by a uniform family of circuits, each using at most a polynomial (in the input size) number of gates - also with at least a $2/3$ probability of success
\end{lemma}
\begin{definition}
    \label{def:BQPvia circuits}
    The complexity class \underline{BQP} is the set of all decision problems that can be decided by a uniform family~\footnote{meaning each circuit of a given size can be constructed in polynomial time by a classical Turing machine.} of quantum circuit, each with depth and number of gates bounded by the size of the input to the circuit.
\end{definition}
Therefore, Lemma~\ref{lem:quantumTuringToCircuits} can equivalently be stated as Def.~\ref{def:BQPviaQuantumTuring} and Def.~\ref{def:BQPvia circuits} define the same set of decision problems. This enables us - throughout the rest of this paper - to consider quantum computations entirely in the form of quantum circuits.
There is then one more definition - Def.~\ref{def:CDevice} - that gives useful notation for how I consider the computations in various complexity classes.
\begin{definition}
    \label{def:CDevice}
    For a complexity class, $C$, a \underline{$C$-device} is anything that can decide all problems in $C$. A \underline{$C$-algorithm} is an algorithm that can be run on a $C$-device, and saying a problem can be decided with \underline{$C$-resources} or an algorithm can run with $C$-resources is equivalent to saying the problem can be decided with a $C$-device and/or using a $C$-algorithm.
\end{definition}

\section{Results}
\subsection{Pre-liminary and Preparatory Results}
The first steps in presenting my results is to establish the foundational tools used in deriving the main results.
\begin{definition}
    Let $\mathcal{L}$ be a decision promise problem.
$\mathcal{L}$ is \underline{$C$-paddable} if, and only if, there exists two functions: $\textit{Pad}_{\mathcal{L}}$, $\textit{Dec}_{\mathcal{L}}$, each of which are $C$-computable. It is also required that $\forall x \in \mathbb{P}, y \in \Sigma^*$:
\begin{align}
    &\textit{Pad}_{\mathcal{L}}(x, y) \in \mathcal{L}
    \iff
    x \in \mathcal{L}\\
    &\textit{Dec}_{\mathcal{L}} \big( \textit{Pad}_{\mathcal{L}}(x, y) \big) = y.
\end{align}
\end{definition}

\begin{definition}
    \label{def:definingBiImmunity}
    A language $\mathcal{L}$ is \underline{$\mathcal{C}$-bi-immune} if neither $\mathcal{L}$ or its complement has an infinite subset in $\mathcal{C}$.
\end{definition}
The most important concept covered by Def.~\ref{def:definingBiImmunity} is P-bi-immunity, which is known to be very common~\cite{MAYORDOMO1994487}. This concept of $\mathcal{C}$-bi-immune can be extended to promise problems, as in Def.~\ref{def:definingBiImmunityInPromise}.
\begin{definition}
    \label{def:definingBiImmunityInPromise}
    A language $\mathcal{L}$ is \underline{$\mathcal{C}$-bi-immune} within a promise, $\mathbb{P}$, if neither $\mathcal{L}$ or its complement has an infinite subset in $\mathcal{C} \cap \mathbb{P}$.
\end{definition}
But it is then useful to define the opposite of Def.~\ref{def:definingBiImmunityInPromise}, in Def.~\ref{def:definingNOTBiImmunityInPromise}.
\begin{definition}
    \label{def:definingNOTBiImmunityInPromise}
    A language $\mathcal{L}$ is not \underline{$\mathcal{C}$-bi-immune} within a promise, $\mathbb{P}$, if either $\mathcal{L}$ or its complement has an infinite subset in $\mathcal{C} \cap \mathbb{P}$.
\end{definition}
A related notion is when a given subset - which can also be a promise - is closed under a padding function.
\begin{definition}
    \label{def:closedPadding}
    For any alphabet, $\Sigma$, any subset, $\mathcal{S}$, of $\Sigma^*$ is said to be \underline{closed under padding} - with respect to a given padding function, $\textit{Pad}_{\mathcal{L}}: \Sigma^* \rightarrow \Sigma^*$, hence I may also say it is closed under the padding function - if $\forall x \in \mathcal{S}$,
    \begin{align}
        x \in \mathcal{S} \Rightarrow \forall y \in \Sigma^*, \textit{Pad}_{\mathcal{L}}(x, y) \in \mathcal{S}.
    \end{align}
\end{definition}
\subsection{Main Result}
Before considering P-bi-immunity, I start more general and consider bi-immunity for more general complexity classes.
\begin{lemma}
    \label{lem:paddMeansNoCBiImmune}
    If a non-trivial~\footnote{Meaning that the correct solution is not always either reject or accept for all instances. A language is said to be trivial if and only if all possible inputs - i.e. those inputs in the promise - are either all in the language or all not in the language.} decidable promise problem, $\mathcal{L} \subseteq \Sigma^*$ (with promise $\mathbb{P} \subseteq \Sigma^*$), is $C$-paddable, such that the promise of the problem is closed under the padding function, it cannot be $C$-bi-immune within the promise.
\end{lemma}
\begin{proof}{[Following the approach of Ref.~\cite[Lemma.~7.26]{greatTextbook}]}
As, by assumptiom, $\mathcal{L}$ is $C$-paddable, there exists two functions: $\textit{Pad}_{\mathcal{L}}$ and $\textit{Dec}_{\mathcal{L}}$ (each of which are $C$-computable), such that, $\forall x \in \mathbb{P}, y \in \Sigma^*$:
\begin{align}
    &\textit{Pad}_{\mathcal{L}}(x, y) \in \mathcal{L}
    \iff
    x \in \mathcal{L}\\
    &\textit{Dec}_{\mathcal{L}} \big( \textit{Pad}_{\mathcal{L}}(x, y) \big) = y.
\end{align}
For any fixed $x \in \mathbb{P}$, I then aim to show that $\textit{Pad}_{\mathcal{L}}(x, y)$ can define a one-to-infinite $C$-computable invertable (with $C$-resources) mapping from $\mathcal{L}$ to itself. To do this define a set that this function will map a given  $x \in \mathbb{P}$ to:
\begin{align}
    \mathbb{S}_x &=
    \bigg\{ \textit{Pad}_{\mathcal{L}}(x, y) \big \vert \text{ } y \in \Sigma^*  \bigg\}.
\end{align}
The important features of the $\mathbb{S}_x$ are that:
\begin{enumerate}
    \item As there are infinite possible choices of $y \in \Sigma^*$, $\forall x \in \mathbb{P}$, $\mathbb{S}_x$ is infinite.
    \item $\forall x \in \mathbb{P}$, as the promise is assumed to be closed under the padding function, $\mathbb{S}_x \subseteq \mathbb{P}$. 
\end{enumerate}
Then the inverse of the function, $f: \mathbb{P} \rightarrow \mathbb{P}$, mapping $x$ to $\mathbb{S}_x$ is infinite-to-one.
However, the central - and \emph{most} important, for my purposes in this paper - property of the $\mathbb{S}_x$ is that, $\forall z \in \Sigma^*$, membership of any specified $\mathbb{S}_x$ can always be decided by a $C$-device. This can be done by applying $f^{-1}: \mathbb{P} \rightarrow \mathbb{P}$ to $z$ and checking if the output is $x$. The full algorithm to calculate $x$ from $\textit{Pad}_{\mathcal{L}}(x, y)$ is a little long so it is easier to check if a given $x \in \mathbb{P}$ is the output.
For a given $\textit{Pad}_{\mathcal{L}}(x, y)$, $\textit{Dec}_{\mathcal{L}} \big( \textit{Pad}_{\mathcal{L}}(x, y) \big) = y$ is much easier to obtain - and can be done with $C$-resources. Likewise $\textit{Pad}_{\mathcal{L}}(x, \textit{Dec}_{\mathcal{L}} \big( \textit{Pad}_{\mathcal{L}}(x, y) \big))$ can always be computed with $C$-resources. Then, $\forall x \in \mathbb{P}$, membership of $\mathbb{S}_x$ can always be checked with $C$-resources, as $\textit{Pad}_{\mathcal{L}}$ and $\textit{Dec}_{\mathcal{L}}$ are $C$-computable, and $\forall z \in \mathbb{P}$:
\begin{align}
    z \in \mathbb{S}_x
    & \iff
    \label{checkCondition}
    z = \textit{Pad}_{\mathcal{L}}(x, \textit{Dec}_{\mathcal{L}} \big(z\big)).
\end{align} The truth of Eqn.~\ref{checkCondition} can be checked by a $C$-device, therefore, $\forall z \in \mathbb{S}_x$, $z \in \mathcal{L}$ can be decided by a $C$-device.
For any $x \in \mathbb{P}$, $\mathbb{S}_x$ is both infinite and a subset of the promise, $\mathbb{P}$.
As padding does not affect membership of the corresponding paddable language every element of $\mathbb{S}_x$ is in the paddable language if and only if the corresponding $x \in \Sigma^*$ is.
I can then define an algorithm that decides if $x \in \mathcal{L}$, using $C$-resources, for an infinite number of instances $x \in \mathbb{P} \cap \mathcal{L}$ and an infinite number of instances $x \in \mathbb{P} \cap \mathcal{L}^c$.
  To do this, first define two possible inputs:
\begin{enumerate}
    \item $x^{\in} \in \mathbb{P} \cap \mathcal{L}$
    \item $x^{\not \in} \in \mathbb{P} \cap \mathcal{L}^c$
\end{enumerate} where $\mathcal{L}^c$ denotes the subset of $\mathbb{P}$ not in $\mathcal{L}$. Both of these exists as, by assumption, the promise problem being considered is non-trivial. Therefore, using the method in Eqn.~\ref{checkCondition} there exists a $C$-algorithm that, $\forall x \in \mathbb{P}$, can decide if $x$ is in $\mathbb{S}_{x^{\in}}$ and, separately, a $C$-algorithm that, $\forall x \in \mathbb{P}$ can decide if $x$ is in $\mathbb{S}_{x^{\not \in}}$. Label these algorithms $\bold{A}_{\in}$ and $\bold{A}_{\not \in}$, respectively.
These are then used in the below Algorithm~\ref{thread2Alg}, which can decide if any element of an infinite and non-trivial subset of $\mathbb{P}$ is in $\mathcal{L}$, with $C$-resources.
\begin{algorithm}[H]
$\bold{Input:}$  $x \in \Sigma^*$, the string to decide if in $\mathcal{L}$
\begin{enumerate}
    \item  \If{\big( $\bold{A}_{\in}(x)$ == ACCEPT \big)}
{
\begin{enumerate}
 \item Result = ACCEPT
\end{enumerate}
}
\item \Else
{
\begin{enumerate}
    \item \If{\big( $\bold{A}_{\not \in}(x)$ == ACCEPT \big)}
{
\begin{enumerate}
    \item Result = REJECT
\end{enumerate}
}
\end{enumerate}
}
\end{enumerate}
$\bold{Return}:$ Result
 \caption{ Algorithm for Deciding if Select Elements of $\mathbb{P}$ are in $\mathcal{L}$
 \label{thread2Alg}} 
\end{algorithm}
Hence, for $\big \vert \mathbb{S}_{x^{\in}} \big \vert$ instances, Algorithm \ref{thread2Alg} decides $x \in \mathcal{L}$ with $C$-resources, as it terminates after step 1a; for $\big \vert  \mathbb{S}_{x^{\not \in}} \big \vert$ instances, Algorithm \ref{thread2Alg} decides $x \not \in \mathcal{L}$ with $C$-resources, as it terminates after step 2 a i; and for all other instances in $\mathbb{P}$, Algorithm \ref{thread2Alg} successfully decides membership of $\mathcal{L}$.
As both $\mathbb{S}_{x^{\in}}$ and $\mathbb{S}_{x^{\not \in}}$ are infinite, this precludes $\mathcal{L}$ from being $C$-bi-immune within $\mathbb{P}$.
\end{proof}
Lemma~\ref{lem:paddMeansNoCBiImmune} was a lemma that applies in the abstract, to a general problem, but it becomes relevant to quantum computing via the below Lemma~\ref{lem:circuitPaddability}.
\begin{lemma}
    \label{lem:circuitPaddability}
    Any decision problem where the input is a circuit - however encoded or formatted, with whatever promise - such that it is possible to express the identity on a single qubit in at least two - efficiently distinguishable - ways, at any point in the circuit, as many times as is desired, such that the circuit remains in the promise; is classically paddable, with the required time to do so scaling linearly in the length of the encoded message and the length of the ways to express the identity on a single qubit.
\end{lemma}
\begin{proof}
    Given a circuit, it is possible to establish a total ordering on the gates of a circuit. This starts by associating each gate with a single qubit: for single-qubit gates, this is easy; for $Q$-qubit gates (where $Q \geq 2$), associate the gate with the top-most qubit (defining top-most however you like, provided it imposes a total ordering on the qubits) that the gate acts on.
    Then, impose an ordering on the qubits. As the set of qubits is finite, imposing a total ordering is always efficiently possible. 
    A gate precedes another in this ordering if the first gate is associated with a qubit (as defined above) that precedes the associated qubit of the second gate in the ordering on qubits just defined.
    If two gates are associated with the same qubit, the precedence between them is decided by which acts first in the circuit.

    With an ordering on the gates established, I can assign a total order to an efficiently identifiable subset of spaces in the circuit. For each gate in circuit, the space immediately preceding it is in the totally ordered subset of spaces and one such space precedes another if the gate associated with the first precedes the other in the ordering on gates previously established.
    This ordering on the subset of spaces in the circuit can be easily converted to an indexing of that same subset.
    This subset can be extended to be an infinite ordered set by considering the end of the circuit - immediately before the measurements - to consist of an infinite number of spaces. These additional spaces are indexed starting at the top-most qubit and moving across qubits before moving along a qubit (i.e. in time).\\

    The padding function works by - after using the decoding function to detect and remove any already-present message (that is either deliberately or accidentally there) - first encoding the message to be added to the circuit in binary, then visiting each site of the above defined indexed set of spaces - in order of the indexing. At the $j$th space, if the $j$th bit of the message to encode is a $1$ add one way of implementing the identity - as is assumed to exist - into the $j$th space, if the $j$th bit of the message to encode is a $0$, add the other way of implementing the zero - also assumed to exist. As this effectively only adds the identity to the circuit, it does not affect the outcome of the circuit, and hence does not affect membership of a language.

    To recover the encoded message, as the ways of encoding the identity are known and efficiently distinguishable, the circuit can just be scanned, efficiently, to find them, and the bits of the encoded message recover their order from the ordering on the spaces they occupy. The message can then be decoded from binary.

    Therefore the language under consideration has both a valid encoding and decoding function, so is paddable. Notably, both the encoding and decoding can be performed classically and in time scaling  linearly in the length of the encoded message and the length of the ways to express the identity on a single qubit.
\end{proof}
Lemma~\ref{lem:circuitPaddability} has now provided a guide to when promises of sDCS are not P-bi-immune. I can then start applying it to specific promises - expressed less formally - of sDCS, in the main theorem of this manuscript: Theorem~\ref{mainTheorem}.
\begin{theorem}
    \label{mainTheorem}
    sDCS is not P-bi-immune regardless of any restrictions on:
    \begin{enumerate}
        \item which single-qubit gates can be applied immediately after state prep, that essentially encode the inputs (and would be inputs in DCS), provided it does not interfere with the encoding
        \item which single-qubit gates can be applied immediately before measurement, that essentially encode which measurements are applied (and would be measurements in DCS, e.g. what basis they are in, how many there are, or any arbitrarily complex rule), provided it does not interfere with the encoding
        \item the number of qubits (provided no upper bound is placed on them), (e.g. must be even, must be prime)
        \item the gateset (e.g. restrictions on which gates can be used restrictions on the number of a certain gate used, the universality of the gateset)
        \item the connectivity of the device (i.e. what qubits multi-qubit gates can go between) or the order in which different two qubit gates can apply
        \item  on particular combinations of gates being applied, provided it does not preclude the encodings of the identity
    \end{enumerate}
\end{theorem}
\begin{proof}
    Any of the enumerated restrictions create a promise on the inputs to sDCS that meets the requirements of Lemma~\ref{lem:circuitPaddability}. Hence using the padding function established by Lemma~\ref{lem:circuitPaddability}, the requirements of Lemma~\ref{lem:paddMeansNoCBiImmune} are met. Therefore, the promise problem corresponding to sDCS with any of the enumerated restrictions is prohibited from being P-bi-immune.
\end{proof}
\begin{corollary}
    \label{cor:DCSnotImmune}
    DCS is not P-bi-immune regardless of any restrictions on:
    \begin{enumerate}
        \item the inputs to the circuit (e.g. the weight of the binary input string, any requirement on the input string itself)
        \item the measurements applied ( e.g. what basis they are in, how many there are, or any arbitrarily complex rule)
        \item the number of qubits (provided no upper bound is placed on them), (e.g. must be even, must be prime)
        \item the gateset (e.g. restrictions on which gates can be used restrictions on the number of a certain gate used, the universality of the gateset)
        \item the connectivity of the device (i.e. what qubits multi-qubit gates can go between) or the order in which different two qubit gates can apply
        \item on particular combinations of gates being applied, provided it does not preclude the encodings of the identity 
    \end{enumerate}
\end{corollary}
After the main results - Theorem~\ref{mainTheorem} and Corollary~\ref{cor:DCSnotImmune} - have been presented, I also present a lemma that allows us to potentially extend the results of this paper: Lemma~\ref{lem:notPnotC}.
\begin{lemma}
    \label{lem:notPnotC}
    If a promise problem is not P-bi-immune, then it is not $C$-bi-immune for any complexity class containing P.
\end{lemma}
\begin{proof}
    Let $\Sigma$ be the alphabet for the promise problem and $\mathcal{L}$ be the language capturing the promise problem with $\mathbb{P} \subset \Sigma^*$ being the promise. If $\mathcal{L}$ is not P-bi-immune, relative to the promise $\mathbb{P}$, then there exists some infinite subset of $\mathbb{P}$, $\mathcal{S}$, such that every word in $\mathcal{S}$ can be decided in polynomial time. Therefore, every word in $\mathcal{S}$ can be decided with $C$-resources if $P \subset C$. Therefore, the given promise problem is not $C$-bi-immune.   
\end{proof}

\subsection{Additional Properties of sDCS}
In addition to the main results of this paper, we have also presented a decision problem, sDCS. In this section, I investigate further properties of sDCS. These properties are concerned - in part - with reductions between languages.
\begin{definition}
    \label{def:reductions}
    A \underline{reduction} an instance, $x \in \Sigma^*$, of a decision problem defined by a language, $\mathcal{L}_A$, to an instance of another decision problem, characterised by the language $\mathcal{L}_B$ is a mapping, $\mathcal{R}: \Sigma^* \rightarrow \Sigma^*$, such that:
    \begin{align}
        x \in \mathcal{L}_A \iff \mathcal{R} (x) \in \mathcal{L}_B
    \end{align}
\end{definition}
To ease the presentation of this next section, I denote language $A$ being many-to-one reducible to language $B$ using the resources of complexity class $C$ by $A \leq^C_m B$ and, if this reduction is also one-one, I write  $A \leq^C_{1} B$. Then I define a notion very similar to padability but with less stringent conditions, in Def.~\ref{def:weakPadd}. 
\begin{definition}
\label{def:weakPadd}
    A language, $\mathcal{L}_{weak} \subseteq \Sigma^*$, is weakly $C$-paddable if and only if:
    \begin{align}
        \mathcal{L}_{weak} \times \Sigma^*
        \leq^C_{1} \mathcal{L}_{weak}.
    \end{align}
\end{definition}

\begin{lemma}{[Based on Ref.~\cite[Lemma.~7.33]{greatTextbook}]}
    \label{weakyMeans11}
    For any weakly $C$-paddable language, $\mathcal{L}_{\textit{pad}}$, and any language, $\mathcal{L}_r$,
    \begin{align}
        \mathcal{L}_r \leq^C_m  \mathcal{L}_{\textit{pad}}
        \iff
        \mathcal{L}_r \leq^C_{1} \mathcal{L}_{\textit{pad}}.
    \end{align}
\end{lemma}
\begin{proof}
    The $\Leftarrow$ direction follows from one-one reductions being valid many-to-one reductions.\\
    I then turn to the $\Rightarrow$ direction. As $\mathcal{L}_{\textit{pad}}$ is weakly $C$-paddable, $\exists \bold{f}: \Sigma^* \rightarrow \Sigma^*$ such that:
    \begin{align}
        \mathcal{L}_{\textit{pad}} \times \Sigma^* \leq^{C}_{1} \mathcal{L}_{\textit{pad}}
        \textit{ via } \bold{f}.
    \end{align}
    Assuming $\mathcal{L}_r \leq^C_m  \mathcal{L}_{\textit{pad}}$, $\exists \bold{g}: \Sigma^* \rightarrow \Sigma^*$ such that:
    \begin{align}
        \mathcal{L}_{r} \leq^{C}_{m} \mathcal{L}_{\textit{pad}}
        \textit{ via } \bold{g}.
    \end{align}
    I can then define a $\leq^C_{1}$-reduction from $\mathcal{L}_{r}$ to $\mathcal{L}_{\textit{pad}}$, using the mapping, $\bold{h}: \Sigma^* \rightarrow \Sigma^*$, $\forall x \in \Sigma^*$, which I define as:
    \begin{align}
        \bold{h}(x) 
        &=
        \bold{f} \big( \bold{g}(x), x \big).
    \end{align}
    $\bold{h}$ is $C$-computable as it is the composition of $C$-computations. It is also one-one as: assume, for the sake of contradiction $\exists x \not = y \in \Sigma^*$,
    \begin{align}
        \label{eqn:showingh11}
        \bold{h}(x) = \bold{h}(y)
        \iff \bold{f} \big( \bold{g}(x), x \big) = \bold{f} \big( \bold{g}(y), y \big)
    \end{align}
    As $\bold{f}$ is one-one, Eqn.~\ref{eqn:showingh11} implies:
    \begin{align}
         \bold{g}(x) =  \bold{g}(y) \text{ and } x = y \Rightarrow \bot.
    \end{align}
    Therefore, $\bold{h}$ must be one-one. Combining this with $\bold{h}$ being $C$-computable gives:
    \begin{align}
        \mathcal{L}_r \leq^C_{1} \mathcal{L}_{\textit{pad}}.
    \end{align}
\end{proof}
With the necessary preparations completed, I turn to applying them to sDCS, in Theorem~\ref{thm:11ReducTosDCS}.
\begin{theorem}
    \label{thm:11ReducTosDCS}
    Every language in BQP has a one-one polynomial time reduction to sDCS.
\end{theorem}
\begin{proof}
    As sDCS is paddable, via Theorem~\ref{mainTheorem}, it is weakly paddable. As sDCS is complete for promiseBQP, $\forall \mathcal{L}_{bqp} \in$ promiseBQP, $\mathcal{L}_{bqp} \leq^{P}_{m}$ sDCS.
    Therefore, via Lemma \ref{weakyMeans11}, $\forall \mathcal{L}_{bqp} \in$ promiseBQP, $\mathcal{L}_{bqp} \leq^{P}_{1}$ sDCS.
\end{proof}
Theorem~\ref{thm:11ReducTosDCS} then provides a corollary that takes a step back and looks at Theorem~\ref{thm:11ReducTosDCS} from a more complexity theoretic perspective. But it requires a definition first.
\begin{definition}
    A language, $\mathcal{L}_A$, is \underline{complete} for a complexity class if every language in the complexity class can be reduced, in time bounded by a polynomial function of the input string's length, to $\mathcal{L}_A$.
\end{definition}
I can then present the final result of this paper, Corollary~\ref{cor:BQPComplete11Reducible}.
\begin{corollary}
    \label{cor:BQPComplete11Reducible}
    There exists a (promise) BQP-Complete language that every language in (promise) BQP is one-one polynomial time reducible to and is not P-bi-immune for infinitely many promises. 
\end{corollary}
\begin{proof}
    By Theorem~\ref{thm:11ReducTosDCS}, all (promise) BQP languages are one-one reducible to sDCS. Then Theorem~\ref{mainTheorem} implies that this same problem (sDCS) is not P-bi-immune for infinitely many promises.
\end{proof}

\section{Discussion}
In this paper, I have examined - via a defined problem, sDCS - the P-bi-immunity of the simulation of decision problems in quantum computing. This has given an insight into the relation between the complexity classes P and BQP, and showed that there free infinitely many promises with an infinitely large subset of problem instances that can be decided classically in polynomial time. 
Further to this, I also showed that the problem used to examine P-bi-immunity also has the property that every language in (promise) BQP has a one-one reduction to it, which may be enabling in future studies.

I have listed out - in Theorem~\ref{mainTheorem} - several aspects of quantum circuits that can be restricted arbitrarily without affecting the ability to pad the input circuits, and hence these aspects can be restricted arbitrarily without removing the infinite set of classically solvable instances. However, these are likely not all the restrictions what leave infinitely many classically solvable instances. It therefore remains for future work to identify further such restrictions that do not make the problem P-bi-immune, and - perhaps - characterize the infinite classically decidable instances that result.

\section{Acknowledgements}
This work was supported, in part, 
by a EPSRC IAA grant (G.PXAD.0702.EXP) and the UKRI ExCALIBUR project QEVEC (EP/W00772X/2). I would also like to thank Katarzyna Macieszczak and Viv Kendon for comments that inspired this project. 

\bibliography{References}



\end{document}